\documentclass[12pt]{article}
\usepackage{amssymb}
\usepackage{multirow}
\usepackage{color}
\usepackage{amsthm}
\usepackage{epsfig,amssymb,amsfonts,amsmath,graphicx,dsfont,cite,xfrac}
\usepackage{authblk}
\usepackage{subcaption}
\definecolor{mygray}{gray}{0.5}

\usepackage{cite}
\usepackage[colorlinks=true,linkcolor=blue,citecolor=red]{hyperref}

\newtheorem{theorem}{Proposition}

\title{Entanglement Generation through Coherent and Non-Coherent Control}
\author[${1}$]{Marco Enr\'iquez}
\author[${1}$]{Francisco Delgado}

\affil[${1}$]{\footnotesize Tecnologico de Monterrey, School of Engineering and Sciences}

\begin{document}
\date{}
\maketitle

\begin{abstract}
The controlled generation of quantum entanglement from separable states remains a central challenge in quantum information science. Here, we investigate entanglement generation using two related control paradigms: coherent path superposition of local unitary operations and stochastic implementations of Pauli channels under coherent control. We show that entangled states belonging to the Bell, GHZ and W classes, can be deterministically generated from fully separable inputs by coherently superposing alternative sets of local unitary transformations. Conditions on the local operators for entanglement generation are derived, and the resulting states are shown to be locally unitary equivalent to standard multipartite entangled states. We further extend the analysis to noisy scenarios, where separable mixed states evolve through pairs of Pauli channels arranged in path-superposition and indefinite causal order configurations. Closed-form expressions for the output states are obtained, and entanglement is quantified using concurrence. By exploring representative channel families across their parameter space, we identify regimes where stochastic entanglement emerges, determine the associated success probabilities, and characterize trade-offs between entanglement and purity.
\end{abstract}

\thispagestyle{empty}

\section*{Introduction}
Quantum entanglement plays an essential role in the development of new technologies \cite{nielsen1}. It enables protocols such as quantum teleportation, quantum cryptography, and dense coding, to mention some. At the heart of this distinctive quantum resource, two closely related open problems lie: the characterization of entanglement and its reliable generation in realistic physical systems.

Entanglement characterization itself is a nontrivial task. Determining whether a given quantum state is entangled can range from straightforward to computationally intractable, depending on the dimensionality of the system. Different entanglement measures have been shown to provide partial solutions, yet there is no universal and efficient method for arbitrary states \cite{bengtsson2017}. In parallel to characterization, the controlled generation of entanglement is an active area of research. Various strategies have been explored, including dissipative engineering \cite{yang2022}, adiabatic state preparation \cite{wu2017}, and the use of non-classical light-matter interactions \cite{bhattacharya2023, lamprou2024}. Unlike pure states, mixed states require more sophisticated techniques, such as the Peres-Horodecki positive partial transpose (PPT) criterion, convex roof constructions, and measures such as formation entanglement or negativity \cite{gour2025}. Mixed-state entanglement is also known to exhibit subtle and counterintuitive phenomena, including bound entanglement, where states are entangled, yet they cannot be distilled into pure entangled pairs using local operations and classical communication \cite{Horodecki1998,Horodecki2000,Horodecki2005}. Recent work has shown that entanglement in mixed states can be generated and enhanced through mechanisms absent in pure state settings \cite{yamasaki2022, rinp2022, moharramipour2024}. Moreover, far from the use of entangling operations, the introduction of a control system that rules these operations implies that entanglement can be generated at distance by measuring such a control system.

In fact, recent developments in quantum control have shown that operations need not be applied to a fixed classical structure. At the core of these new techniques, the concept of quantum indefiniteness is essential. A type of quantum indefiniteness, called an indefinite causal order (ICO), can be created combining the action of two quantum channels in a superposition of alternative orders \cite{chiribella2013, rubino2017}. Such ICO processes have been shown to provide advantages in communication \cite{Ebler2018, Procopio2019}, channel discrimination \cite{Bavaresco2021}, and parameter estimation \cite{Zhao2020, DelSanto2024}. An alternative control technique that shows quantum indefiniteness is the path superposition paradigm (PS), in which a quantum control system coherently selects between alternative operations arranged along different branches \cite{Abbott2018,Chiribella2019}. This kind of indefiniteness has been shown to bring about some new advantages in quantum communications \cite{Rubino2021}, quantum parametric estimation \cite{Delgado2023}, and quantum teleportation \cite{Mondal2025}. It also provides novel features in the manipulation of a two-level atom interacting with two quantized fields of the quantum cavity \cite{CastanosCervantes2024}.

However, the generation of entanglement using quantum communication architectures has evolved in some major theoretical directions. The more basic approach is the quantum circuit-based methodology, based on the application of unitary entanglement operations, commonly between qubit pairs, but as a series of steps to produce more extended entanglement or entanglement purification. Another approach is the use of measurement-induced entanglement (weak or still projective) \cite{sorensen2003, kim2012}, particularly by using ancilla subsystems to control the entanglement \cite{white2016, Grimaudo2020}, then finally measuring them to produce certain partial projections on the main systems spaces in the entangled states. This controlled structure implies the possibility of producing entanglement in distant systems via a previous basic entanglement, but then the possibility of generating a secondary entanglement through passive or imperfect communication channels where input states lie. Although possibly those methodologies could become stochastic, the proper control introduced could still be arranged to fit almost certain procedures. The interest in this entanglement generation lies in quantum processing resourcing, particularly distributed quantum computing. Also, in quantum memories and repeaters, together with quantum sensing networks and quantum internet. The entangling architectures seek to generate entanglement efficiently and then distribute it across distant nodes of the circuit.

In this work, we investigate entanglement generation arising from two closely related architectures: (i) coherent path superposition of local unitary operations, and (ii) stochastic arrangements of Pauli channels under path superposition. First, we demonstrate that maximally entangled states, including Bell, GHZ, and W states, can be generated from fully separable inputs by coherently superposing alternative sets of local unitary transformations. Conditions for maximal entanglement generation are derived, and the resulting states are shown to be locally unitary equivalent to the standard representatives of these entanglement classes. Second, we extend the analysis to noisy scenarios in which separable mixed states are processed through pairs of Pauli channels arranged under coherent control. We obtain closed analytical expressions for the output states and quantify the generated entanglement using the concurrence of two qubits. By exploring emblematic families of channels across their full parametric space, we identify the regions where stochastic entanglement emerges, analyze the associated success probabilities, and characterize purity–entanglement trade-offs.

The paper is structured as follows. In Section \ref{PSU}, we introduce a scheme for generating multipartite entangled states based on the coherent path superposition of local unitary operations, demonstrating how entanglement can emerge from fully separable inputs after a projective measurement on a control system, and deriving the necessary and sufficient conditions for obtaining maximally entangled Bell, GHZ, and W states. In Section \ref{StochasticSection}, we extend this framework to noisy scenarios by considering pairs of Pauli channels arranged under path superposition and indefinite causal order, and we derive general analytical expressions for the corresponding output states, highlighting their dependence on channel parameters and control configurations. In Section \ref{EntanglementSection}, we analyze the entanglement properties of these states using concurrence and purity as quantitative measures, examining several representative families of channels to identify regimes where stochastic entanglement arises, as well as the associated success probabilities and trade-offs with decoherence. Additionally, we provide a global characterization of entanglement generation across the full parameter space of Pauli channels, establishing the capabilities and limitations of coherent and non-coherent control strategies for entanglement generation. Finally, Section \ref{conclu} is devoted to provide some conclusions and perspectives of our work.

\section{$n$-qubit GHZ and W states generation via path superposition}\label{PSU}
In this section, we discuss a scheme for generating maximally entangled states. We begin with a two-qubit system prepared in a fully separable state. A control qubit coherently selects between two local operations $U_0\otimes U_1$ and $\widetilde U_0\otimes \widetilde U_1$, which act on the initial state. Since both operations are local, they cannot generate entanglement. However, we show that when these two alternatives are placed in a coherent superposition of paths, entanglement can emerge after a projective measurement on the control system. The operator implementing this controlled transformation can be written as
\begin{equation}\label{UST}
    {\cal U}_2=U_0\otimes U_1 \otimes \vert 0\rangle \langle 0\vert_C +\widetilde{U}_0\otimes \widetilde{U}_1 \otimes \vert 1\rangle \langle 1\vert_C,
\end{equation}
and acts in the initial state $\vert \varphi_0\rangle\otimes\vert \varphi_1\rangle\otimes \vert \pm_C \rangle$, where $\vert \pm_C \rangle= (\vert 0_C\rangle\pm\vert 1_C\rangle)/\sqrt 2$ is the control state choosing which local operator acts in the separable state. This operator is depicted in Fig.~\ref{uPathS}(a). The state obtained after the action of (\ref{UST}) is
\begin{equation}
    \vert \psi\rangle_2 = \frac 12(U_0\otimes U_1+\widetilde{U}_0\otimes \widetilde{U}_1)  \vert \phi_0\phi_1\rangle \otimes \vert +_C\rangle+\frac 12(U_0\otimes U_1-\widetilde{U}_0\otimes \widetilde{U}_1) \vert \phi_0\phi_1\rangle \otimes \vert -_C\rangle.
\end{equation}
\begin{figure}[bth]
\centering
\begin{tabular}{ccc}
\includegraphics[scale=0.35]{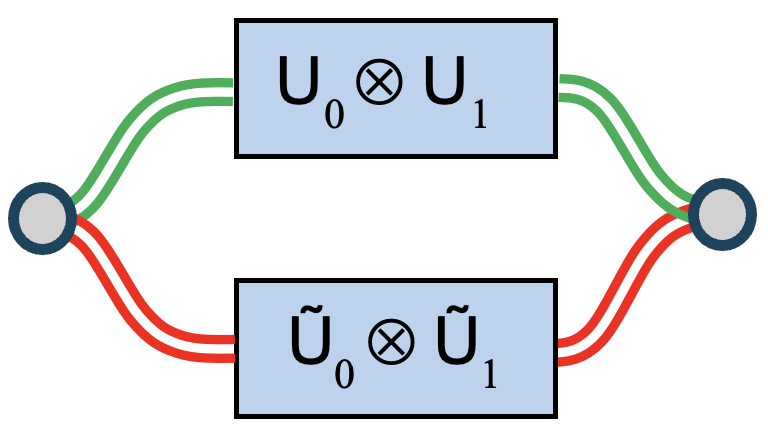} & \hspace{10pt} &\includegraphics[scale=0.3]{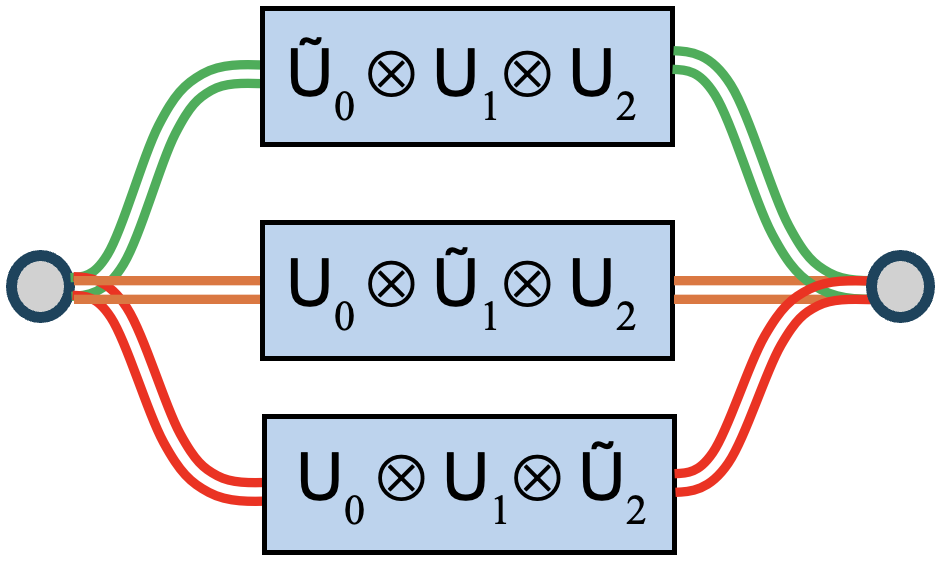}\\
(a) & &(b)
\end{tabular}
\caption{\label{uPathS} Appropriately arranged local unitary operators, coherently superposed along distinct paths, act on an initially fully separable state to generate entanglement. A control system determines which local operator act on the separable state. Panels (a) and (b) schematically illustrate respectively the operators ${\cal U}_2$ in Eq. (\ref{UST}) ${\cal U}_3$ in Eq. (\ref{UST3}).}
\end{figure}
 After a projective measurement on the control basis of $\left\{\vert \pm_C\rangle\right\}$, the resulting two-qubit state reads
\begin{equation}\label{psipm}
    \vert \psi_\pm\rangle_2 =\frac1{\sqrt {N_\pm}} (U_0\otimes U_1\pm \widetilde{U}_0\otimes \widetilde{U}_1)  \vert \phi_0\phi_1\rangle.
\end{equation}
where $N_\pm$ is a normalization constant. We now determine the conditions under which the state (\ref{psipm}) becomes maximally entangled. Without loss of generality, we focus on the state $\lvert \psi_+ \rangle_2$, since it is locally unitary (LU) equivalent to $\lvert \psi_- \rangle_2$ through the transformation $\mathbb{I} \otimes P$, where $P \lvert \phi_1 \rangle = - \lvert \phi_1 \rangle$.

\begin{theorem}[Bell states condition] The state in Eq. (\ref{psipm}) is maximally entangled if and only if the two branches of the superposition are mutually orthogonal, namely
\begin{equation}\label{cond1}
    \langle \phi_k\vert \widetilde{U}_k^\dagger U_k\vert \phi_k\rangle  {\color{black} = 0}, \quad k=0,1.
\end{equation}
\end{theorem}
\begin{proof}
Necessity. To prove this result, note that maximal entanglement requires the two product states appearing in Eq. (\ref{psipm}) to contribute with equal weights and orthogonal support. Since each branch corresponds to a local unitary transformation acting on the initial separable state, orthogonality is achieved precisely when the local overlaps vanish as expressed in Eq. (\ref{cond1}). Under these conditions, the resulting state becomes locally unitary equivalent to a Bell state. 

Sufficiency. If the conditions in Eq. (\ref{cond1}) are satisfied, the two product states appearing in Eq. (\ref{psipm}) are orthogonal. Since the coefficients of both terms are equal to $1/{\sqrt 2}$, Eq. (\ref{psipm}) is already a Schmidt decomposition. Hence the reduced density matrix is maximally mixed, and the generated state is maximally entangled (i.e., locally unitary equivalent to a Bell state).
\end{proof}
As a particular example, choose $U_0 = \widetilde{U}_{\color{black} 1} = \sigma_1$ and $U_1 = \widetilde{U}_{\color{black} 0} = \sigma_3$, and take $\lvert \phi_k \rangle = \sqrt{\alpha}\,\lvert 0 \rangle + \sqrt{1-\alpha}\,\lvert 1 \rangle$, with $\alpha \in [0,1]$. Remarkably, the final state is independent of the parameter $\alpha$ and coincides with the Bell state $\lvert \Phi^+ \rangle=(\vert00\rangle+\vert 11\rangle)/\sqrt 2$.

In Ref.~\cite{koudia2023}, entanglement generation was analyzed through a superposition of causal orders, where the ordering of the local operations constitutes an additional degree of freedom. As a consequence, the entanglement conditions depend on both the properties of the local transformations and the coherence between different causal structures. By contrast, in the path-superposition framework considered here, the ordering of the operations remains fixed and the only source of nonclassicality is the coherent superposition of alternative trajectories. The resulting entanglement conditions therefore acquire a particularly simple form, being determined solely by orthogonality relations between the local branches.

Unlike the conventional quantum circuit used to generate a Bell state, typically based on a fixed sequence of a Hadamard gate followed by a controlled-NOT (CNOT) gate, the procedure introduced in Eq. (\ref{UST}) relies on a fundamentally different mechanism rooted in coherent control of operations rather than direct entangling interactions. In the standard circuit, entanglement is generated deterministically through a nonlocal two-qubit gate (CNOT), which directly couples the qubits and creates correlations via interaction. In contrast, the scheme in Eq. (\ref{UST}) employs only local unitary operations acting independently on each subsystem. By themselves, these operations cannot generate entanglement. The key resource enabling entanglement in this case is the coherent superposition of alternative operational paths, controlled by a control system. The operator ${{\cal U}_2}$ conditionally applies different pairs of local unitaries depending on the state of the control qubit, effectively placing the evolution in a superposition of distinct transformations. Entanglement then emerges only after a projective measurement on the control system, which creates interference between these alternative evolutions. Therefore, unlike the conventional circuit where entanglement arises from explicit two-qubit interactions within a definite causal structure, here it is generated indirectly through quantum interference and post-selection, highlighting a qualitatively different paradigm based on coherent control and measurement.

The previous results can be straightforwardly extended to the $n$-qubit case, yielding the GHZ state by considering the action of the operator
\begin{equation}\label{ghzNq}
    {\cal U}_N=U_0\otimes U_1\otimes \cdots \otimes U_{n-1}\otimes \vert0\rangle\langle0\vert_C+\widetilde U_0\otimes \widetilde U_1\otimes \cdots \otimes \widetilde U_{n-1}\otimes \vert1\rangle\langle1\vert_C.
\end{equation}
Thus, the following Proposition can be set forth
\begin{theorem}[GHZ-states condition]
The coherent superposition of alternative sets of local operations described by Eq. (\ref{ghzNq}) generates $n$-qubit GHZ-type states whenever the corresponding local orthogonality conditions are satisfied for every subsystem.
\end{theorem}

\begin{proof}
The result follows directly from Proposition~1. The orthogonality conditions
$
\langle \varphi_k|\widetilde U_k^\dagger U_k|\varphi_k\rangle=0,
\quad k=0,\ldots,n-1,
$
ensure that the two $n$-qubit product branches appearing in Eq.~(5) are mutually orthogonal. Since both branches contribute with equal amplitudes, the post-selected state is an equal superposition of two orthogonal product states. By means of local unitary transformations, these branches can always be mapped onto $|00\cdots0\rangle$ and $|11\cdots1\rangle$, respectively. Therefore, the generated state is locally unitary equivalent to the $n$-qubit GHZ state.

\end{proof}
 Similarly, the $n$-qubit W state can be obtained. As an illustrative case, we consider $n=3$. Two sets of unitary operators $\left \{ U_k\right \}$ and $\left \{ \widetilde U_k\right \}$ for $k=0,1,2$ can be used to construct the operator
\begin{equation}\label{UST3}
    {\cal U}_3=\widetilde U_0\otimes U_1\otimes U_2\otimes \vert 0\rangle\langle0\vert_C+ U_0\otimes \widetilde U_1\otimes U_2\otimes \vert 1\rangle\langle1\vert_C+U_0\otimes U_1\otimes \widetilde U_2\otimes \vert 2\rangle\langle2\vert_C,
\end{equation}
The action of this operator is depicted in Fig.\ref{uPathS}(b) and acts on the composite state $\vert \phi_0\rangle\otimes \vert \phi_1\rangle \otimes \vert \phi_2\rangle \otimes\vert \psi\rangle_c$, 
where the control qubit is prepared in one of the states of the discrete Fourier transform (DFT) basis:
\begin{equation}\label{DFTbasis}
    \begin{array}{ll}
         \vert +_C\rangle=(\vert 0_C\rangle+\vert 1_C\rangle_c+\vert 2_C\rangle)/\sqrt3,   &
         \vert -_C\rangle=(\vert 0_C\rangle+\omega \vert 1_C\rangle+\omega^2\vert 2_C\rangle)/\sqrt3, \\[1em]
         \vert \times_C \rangle=(\vert 0_C\rangle+\omega^2 \vert 1_C\rangle+\omega \vert 2_C\rangle)/\sqrt3
    \end{array}
\end{equation}
with $\omega=e^{2\pi i/3}$. To simplify the calculations, let us assume that the control state is initially in the state $\vert+_C\rangle$. Indeed,
\begin{equation}\label{psiplus}
\begin{array}{ll}
    \vert \psi\rangle_3 =\displaystyle \frac1{\sqrt 3}\left(\widetilde U_0\otimes U_1\otimes U_2\vert \phi_0\rangle \vert\phi_1\rangle\vert\phi_2\rangle\vert 0\rangle_c+U_0\otimes \widetilde U_1\otimes U_2\vert \phi_0\rangle \vert\phi_1\rangle\vert\phi_2\rangle\vert 1\rangle_c\right.\\[1em]
    
    \hspace*{5cm} \displaystyle\left.+U_0\otimes U_1\otimes \widetilde U_2\vert \phi_0\rangle \vert\phi_1\rangle\vert\phi_2\rangle\vert 2\rangle_c\right).
\end{array}
\end{equation}
A measurement on the basis (\ref{DFTbasis}) would yield
\begin{equation}\label{psi3p}
    \vert\psi_+\rangle_3=\frac1{\sqrt{N_+}}\left(\widetilde U_0\otimes U_1\otimes U_2+U_0\otimes \widetilde U_1\otimes U_2+U_0\otimes U_1\otimes \widetilde U_2\right)\vert \phi_0\rangle \vert\phi_1\rangle\vert\phi_2\rangle,
\end{equation}
\begin{equation}\label{psi3m}
    \vert\psi_-\rangle_3=\frac1{\sqrt{N_-}}\left(\widetilde U_0\otimes U_1\otimes U_2+{\color{black} \omega} U_0\otimes \widetilde U_1\otimes U_2+\omega^{\color{black} 2} U_0\otimes U_1\otimes \widetilde U_2\right)\vert \phi_0\rangle \vert\phi_1\rangle\vert\phi_2\rangle,
\end{equation}
\begin{equation}\label{psi3x}
    \vert\psi_\times\rangle_3=\frac1{\sqrt{N_\times}}\left(\widetilde U_0\otimes U_1\otimes U_2+\omega^{\color{black}  2} U_0\otimes \widetilde U_1\otimes U_2+{\color{black} \omega} U_0\otimes U_1\otimes \widetilde U_2\right)\vert \phi_0\rangle \vert\phi_1\rangle\vert\phi_2\rangle.
\end{equation}
The states given by Eqs. (\ref{psi3m})–(\ref{psi3x}) differ only by relative phase factors inherited from the Fourier basis of the control system. Since these phases can be removed through local phase transformations acting on the individual qubits, the three states are locally unitary equivalent. Consequently, it is sufficient to analyze the state (\ref{psi3p}) without loss of generality.
\begin{theorem}[W-state condition]
 Let the three operational branches generated by ${\cal U}_3$ satisfy the orthogonality conditions
\begin{equation}\label{cond3q}
\langle \phi_k|\tilde U_k^\dagger U_k|\phi_k\rangle=0,
\qquad k=0,1,2.
\end{equation}
Then the post-selected states given by Eqs. (9)–(11) are locally unitary equivalent to the three-qubit W state.
\end{theorem}

\begin{proof}
The conditions (\ref{cond3q}) ensure that the three computational-basis contributions appearing in (\ref{psi3p}) are mutually orthogonal and occur with equal amplitudes. Therefore the state $\vert \psi_+\rangle_3$ can be written in the form
\begin{equation}
|\psi_+\rangle_3=
\frac{1}{\sqrt3}
\left(
|100\rangle
+
e^{i\phi_1}|010\rangle
+
e^{i\phi_2}|001\rangle
\right),
\end{equation}
where the phases depend on the control measurement outcome.
Since these phases can be removed through local phase rotations acting independently on the three qubits, the state is locally unitary equivalent to
\begin{equation}
|W\rangle=
\frac{1}{\sqrt3}
\left(
|100\rangle+|010\rangle+|001\rangle
\right).
\end{equation}
The same argument applies to Eqs. (\ref{psi3m}) and (\ref{psi3x}), which are already locally unitary equivalent to Eq. (\ref{psi3p}). Therefore all three post-selected states belong to the W-entanglement class. 
\end{proof}
The results above show that Bell, GHZ, and W-state generation follow from the same basic mechanism. In all cases, coherent control produces a post-selected superposition of mutually orthogonal product branches with equal amplitudes. Depending on the number and structure of the branches, this superposition becomes locally unitary equivalent to a Bell state, a GHZ-type state, or a W-type state.
Finally we remark that Propositions 2 and 3 demonstrate that the same coherent-control architecture can generate representative states from the two inequivalent classes of genuine three-qubit entanglement, namely the GHZ and W classes. This highlights the versatility of the proposed framework and its ability to generate distinct multipartite entanglement structures through suitable choices of local operations and control measurements.

\section{Stochastic states emerging from a channel arrangement}\label{StochasticSection}

The previous coherent example has shown how the path superposition of unitary operations can produce an entangled state. Can other communication architectures, for instance, indefinite causal order, generate them? can this entanglement be created through non-coherent operations? For this purpose, in this section, we introduce both architectures to investigate if Pauli noise \cite{flammia2020}, as a source of non-coherent operations, is able to produce entanglement. The entanglement analysis will be developed in the next section.

In the current analysis, we will consider a couple of single qubit Pauli channels, each of the form \cite{nielsen1,delgado3}:

\begin{equation} \label{paullichannel}
    \rho_{out} = \Lambda[\rho_{in}] = \sum_{i=0}^3 K_i \rho_{in} K_i^{\dagger} = \sum_{i=0}^3 \alpha_i\sigma_i \rho_{in} \sigma_i^{\dagger} 
\end{equation}

\noindent with $\sigma_i$ the Pauli operators for $i=1,2,3$ and $\sigma_0$ the identity operator. Here, $K_i$ are the corresponding Kraus operators \cite{kraus1} of the operation introduced by the single Pauli channel with $K_i=\sqrt{\alpha_i}\sigma_i, i=0,1,2,3$. Note that $\sum_{i=0}^3 K_i \rho K_i^{\dagger}={\mathbf 1}$ and therefore $\sum_{i=0}^3 \alpha_i=1$.

Here, we will work with the couple of Pauli channels $\Lambda_1(\rho)$ and $\Lambda_2(\rho)$ under two different arrangements: a) a Path Superposition (PS) arrangement and b) an Indefinite Causal Order (ICO) arrangement \cite{chiribella2013}. These arrangements are shown in Figure \ref{Fig1}. Figure \ref{Fig1}(a) shows the PS arrangement; note that each line represents the two involved systems. Similarly, Figure \ref{Fig1}(b) shows an ICO arrangement. In each case, the boxes show a pair of channels $\Lambda_1(\rho_1)$ and $\Lambda_2(\rho_2)$ (the subscript here only refers to two different channels, not to the system being applied) applied in each case to different systems (in that order, first for system 1 and then for system 2), symbolically represented as $\Lambda_1 \otimes \Lambda_2$ or $\Lambda_2 \otimes \Lambda_1$. In each case, the path or causal order is stated by a quantum control system $\vert \Phi_C \rangle = \sqrt{p_0} \vert 0_C \rangle + \sqrt{p_1} \vert 1_C \rangle$, thus if the control state is $\vert 0_C \rangle$ the green path is followed, but if the control state is $\vert 1_C \rangle$ then the red one is followed.

\begin{figure}[tbh]
\centering
    \begin{tabular}{cc}
         \includegraphics[width=0.30\linewidth]{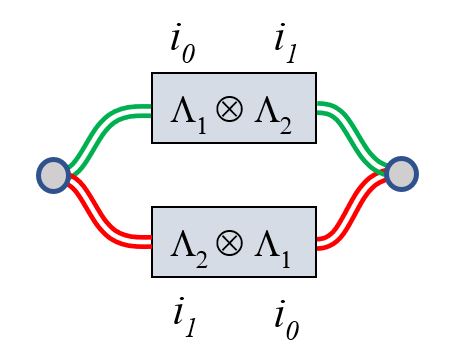}  &
         \raisebox{0.012\height}
         {\includegraphics[width=0.29\linewidth]{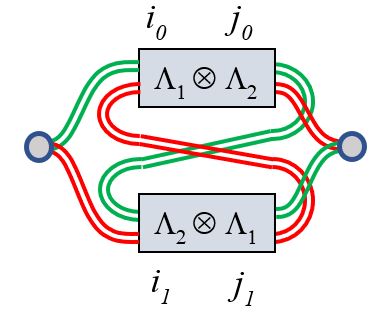}}  \\
         (\textbf{a}) & (\textbf{b})
    \end{tabular}
    \caption{Two-qubit communication architectures implementing a couple of channels, $\Lambda_1$ and $\Lambda_2$:  (\text{A}) Path superposition, and (\textbf{B}) Indefinite Causal Order. }\label{Fig1}
\end{figure}

The output state, for the operations performed by the arrangements of Pauli channels presented before, should be obtained by considering the associated Kraus operators (here, we will restrict the input state as separable and mixed $\rho_1 \otimes \rho_2$):

\begin{eqnarray}
    {\tilde \Lambda}(\rho_1 \otimes \rho_2) &=& \sum_{\{\mathcal C\}} K^{\{\mathcal C\}} \rho_1 \otimes \rho_2 {K^{\{\mathcal C\}}}^\dagger 
\end{eqnarray}

\noindent by denoting in general $\{\mathcal C\}$ the set of general indices with respect to the coherent components of the channels. Such Kraus operators could be written in terms of the form coefficients \cite{Delgado2023} of the arrangement:

\begin{eqnarray}
    K^{\{\mathcal C\}} = \sum_{\alpha \beta \gamma \delta \epsilon} C^{\{\mathcal C\}}_{\alpha \beta \gamma, \delta \epsilon} \vert \alpha_1 \beta_2 \gamma_C \rangle \langle \delta_1 \epsilon_2 \vert
\end{eqnarray}

For both cases considered here, $\{{\mathcal C}\}$ becomes the set $i_0,i_1,j_0,j_1$ for PS and ICO corresponding to each different channel involved. In general, in the following, Greek scripts sum over $0$ and $1$, while Latin scripts sum over $0,...,3$. Note that Figure \ref{Fig1} has maintained the scripts related to each channel to facilitate interpretation. Then, the form coefficients have the following expressions (in this representation, it is not necessary to explicitly include the expression $\vert \Phi_C \rangle$ beyond $p_\gamma$):

\begin{eqnarray}
    C^{i_0,j_0}_{\alpha \beta \gamma, \delta \epsilon} = \sqrt{p_\gamma \alpha^{1+\gamma}_{i_0} \alpha^{1+(\gamma \oplus 1)}_{j_0}} 
    \left( \sigma_{1_{i_{\gamma}}}\right)_{\alpha \delta}\left( \sigma_{2_{j_{\gamma}}}\right)_{\beta \epsilon}
\end{eqnarray}

\noindent for the PS case (Figure \ref{Fig1}(a)), and:

\begin{eqnarray}
    C^{i_0,i_1,j_0,j_1}_{\alpha \beta \gamma, \delta \epsilon} = \sqrt{p_\gamma \alpha^{1}_{i_0} \alpha^{2}_{i_1} \alpha^{2}_{j_0} \alpha^{1}_{j_1}} 
    \left( \sigma_{1_{i_{\gamma \oplus 1}}} \sigma_{1_{i_{\gamma}}}\right)_{\alpha \delta}\left( \sigma_{2_{j_{\gamma \oplus 1}}} \sigma_{2_{j_{\gamma}}}\right)_{\beta \epsilon} 
\end{eqnarray}

\noindent for the ICO case (Figure \ref{Fig1}(b)). There, $\alpha^j_{k}$ corresponds to the $k-$coefficient of the Pauli channel $\Lambda_j$; also, $\sigma_{j_{k}}$, is the $k-$element of the set $\{\sigma_0, \sigma_1, \sigma_2, \sigma_3\}$ (Pauli matrices together with the identity) applied to the subsystem $j$. Subscripts outside the parentheses refer to the correspondent entry within it. Each input state could be written in the Bloch representation $\rho_i = \frac{1}{2}(\sigma_{0_i} + {\vec n}\cdot{\vec \sigma}_i)$, with ${\vec \sigma}=(\sigma_1,\sigma_2,\sigma_3)$, the Pauli matrices vector. However, to simplify our expressions and notation, this representation will be written in a simpler notation $\rho_i = \frac{1}{2} {\vec N}_i \cdot {\vec S}_i, i=1, 2$, with ${\vec N}_i = (1,{\vec n}_i)$ and ${\vec S}_i=(\sigma_{0_i},{\vec \sigma_{i}})$ the $4-$vector extensions for the Bloch vector and the spin vector operator \cite{Delgado2023}, respectively.
Thus, the output state including the control state becomes:

\begin{eqnarray}
    {\tilde \Lambda}(\rho_1 \otimes \rho_2) &=& \frac{1}{4} \sum_{\substack{k,k', \{\mathcal C\} \\ \alpha \beta \gamma \delta \epsilon \\ \alpha' \beta' \gamma ' \delta' \epsilon'}} N_{1_k} N_{2_{k'}} C^{\{\mathcal C\}}_{\alpha \beta \gamma, \delta \epsilon} {C^{\{\mathcal C\}*}_{\alpha' \beta' \gamma', \delta' \epsilon'}} S_{1_{k_{\delta \delta'}}} S_{2_{k'_{\epsilon \epsilon'}}} \vert \alpha \beta \gamma \rangle \langle \alpha' \beta' \gamma' \vert
\end{eqnarray}

\noindent  Concretely, each $N_{j_k}$ is the $k-$component of ${\vec N}_j$ for the subsystem $j$; and $S_{j_{k_{\epsilon \epsilon}}}$ is the $(\epsilon,\epsilon')$ entry of the $k-$component of ${\vec N}_j$ in the subsystem $j$. In addition, considering a stochastic approach where the control state is measured expecting to have an optimal outcome ($\mu=0$) on the basis $\{\vert \psi^\mu_C \rangle = \sum_{i \in \{0,1\}} a_i^\mu \vert i_C \rangle| \mu=0,1\}$, then, the un-normalized output state regarding just the couple of systems:

\begin{eqnarray}
    {\Lambda}(\rho_1 \otimes \rho_2)_{\rm un} &=& \langle \psi_C \vert {\tilde \Lambda}(\rho_1 \otimes \rho_2) \vert \psi_C \rangle \\
    &=& \sum_{\substack{\alpha \beta \\ \alpha' \beta'}} \bigg( \frac{1}{4} \sum_{\substack{k,k', \{\mathcal C\} \\ \gamma \delta \epsilon \\ \gamma' \delta' \epsilon'}} N_{1_k} N_{2_{k'}} C^{\{\mathcal C\}}_{\alpha \beta \gamma, \delta \epsilon} {C^{\{\mathcal C\}*}_{\alpha' \beta' \gamma', \delta' \epsilon'}} S_{1_{k_{\delta \delta'}}} S_{2_{{k'}_{\epsilon \epsilon'}}} {a^\mu_\gamma}^* a^\mu_{\gamma'} \bigg) \vert \alpha \beta \rangle \langle \alpha' \beta' \vert  \\
    &\equiv& \sum_{\substack{\alpha \beta \\ \alpha' \beta'}} \rho_{\alpha \beta, \alpha' \beta'} \vert \alpha \beta \rangle \langle \alpha' \beta' \vert
\end{eqnarray}

\noindent Finally, the normalized output state and its probability of success in the stochastic process are as follows.

\begin{eqnarray}\label{finalexps}
    {\Lambda}(\rho_1 \otimes \rho_2) &=& \frac{{\Lambda}(\rho_1 \otimes \rho_2)_{\rm un}}{P_{\mu=0}} \\
    P_{\mu=0} &=& {\rm Tr}({\Lambda}(\rho_1 \otimes \rho_2)_{\rm un}) \\
    &=& \frac{1}{4} \sum_{\substack{k,k', \{\mathcal C\} \\ \alpha \beta \gamma \delta \epsilon \\ \gamma' \delta' \epsilon'}} N_{1_k} N_{2_{k'}} C^{\{\mathcal C\}*}_{\alpha \beta \gamma, \delta \epsilon} {C^{\{\mathcal C\}}_{\alpha \beta \gamma', \delta' \epsilon'}} S_{1_{k_{\delta \delta'}}} S_{2_{k'_{\epsilon \epsilon'}}} {a^\mu_\gamma}^* a^\mu_{\gamma'}
\end{eqnarray}

Interestingly, we note that both arrangements depend on the same number of parameters (fourteen parameters as a total): $\{\alpha^i_j|i=1,2; j=1,...,3\}$ for the couple of channels $i=1,2$ and in spite of $\alpha^i_0=1-\sum_{j=1}^3 \alpha^i_j$. In addition, $p_0, a_0$ (clearly $p_1=1-p_0, a_1^2=1-a_0^2$), ${\vec n}_i, i=1,2$ for the two separable initial states (six more parameters). 
\section{Entanglement}\label{EntanglementSection}
In the previous sections, we established that coherent control of alternative operational pathways can generate entanglement both in purely unitary scenarios and in noisy Pauli-channel architectures. We now quantify the resulting correlations using concurrence and analyze their interplay with purity and success probability. This analysis allows us to characterize the operational regimes in which coherent and non-coherent control most efficiently generate entanglement.

\subsection{Concurrence}
A widely used two-qubit entanglement measure is the concurrence $\mathcal C$ \cite{HW97, W98}. For a two-qubit system described by a density matrix $\rho$, this is defined as:
\begin{equation}
\mathcal{C}(\rho) = \max(0, \lambda_1 - \lambda_2 - \lambda_3 - \lambda_4)
\end{equation}
where $\lambda_i$ are the eigenvalues, in descending order, of the matrix:
\begin{equation}
R = \sqrt{\sqrt{\rho} (\sigma_2 \otimes \sigma_2) \rho^* (\sigma_2 \otimes \sigma_2) \sqrt{\rho}},
\end{equation}
here, the operation $^*$ represents the complex conjugate. For pure states $|\psi\rangle$, the concurrence reduces to:
\begin{equation}
\mathcal{C}(|\psi\rangle) = 2 |\langle \psi | \sigma_2 \otimes \sigma_2 | \psi^* \rangle|,
\end{equation}
which corresponds to twice the determinant of the coefficient matrix in the computational basis.

\subsection{Concurrence of output states for representative channel families}

\subsubsection{Entanglement generated by coherent control transitioning into non-coherent control}

We analyze the concurrence $\mathcal{C}(\rho_{\rm out})$ with $\rho_{\rm out}=\Lambda(\rho_1 \otimes \rho_2)$ in (\ref{finalexps}), as a function of $p=p_0, q=s_0$ (just $\rho$ in the following, for simplicity). We add a third parameter $s$ that shows the channels involved. In any case, we depart from the pure input state $|0\rangle\otimes|0\rangle.$ For simplicity, we consider the case with $\Lambda_1$ defined by $\alpha^1_0=1$. Thus, by taking $\Lambda_2$ characterized by $\alpha^2_0=1-s, \alpha^2_1=s$. The results are shown in Figure \ref{Fig2}. 

\begin{figure}[bht]
\begin{center}
    \begin{tabular}{cc}
        {\multirow{3}[2]{*}[20mm]{\raisebox{0.5\height}{\includegraphics[width=0.350\linewidth]{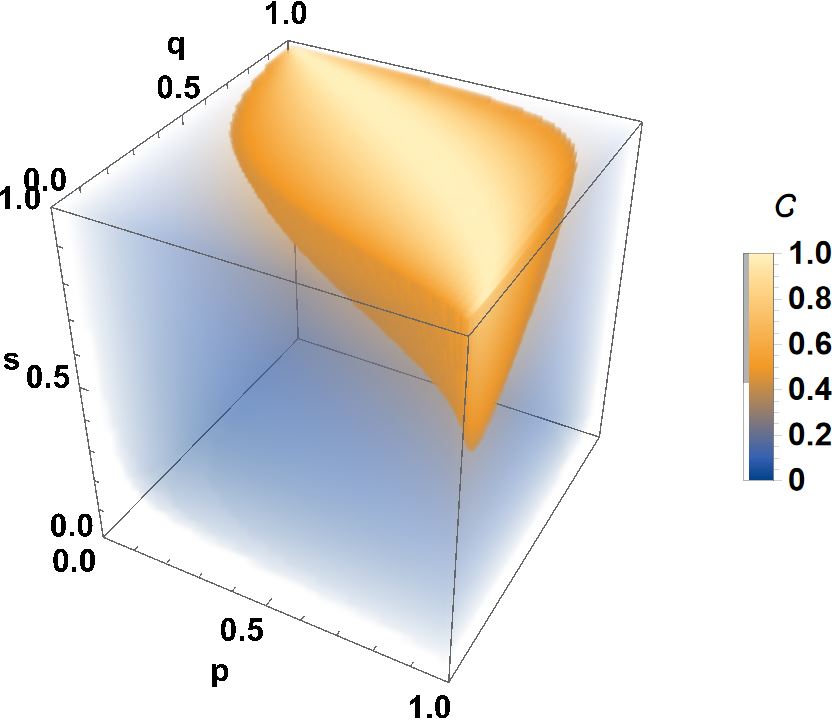}}}}  &
        {\includegraphics[width=0.580\linewidth]{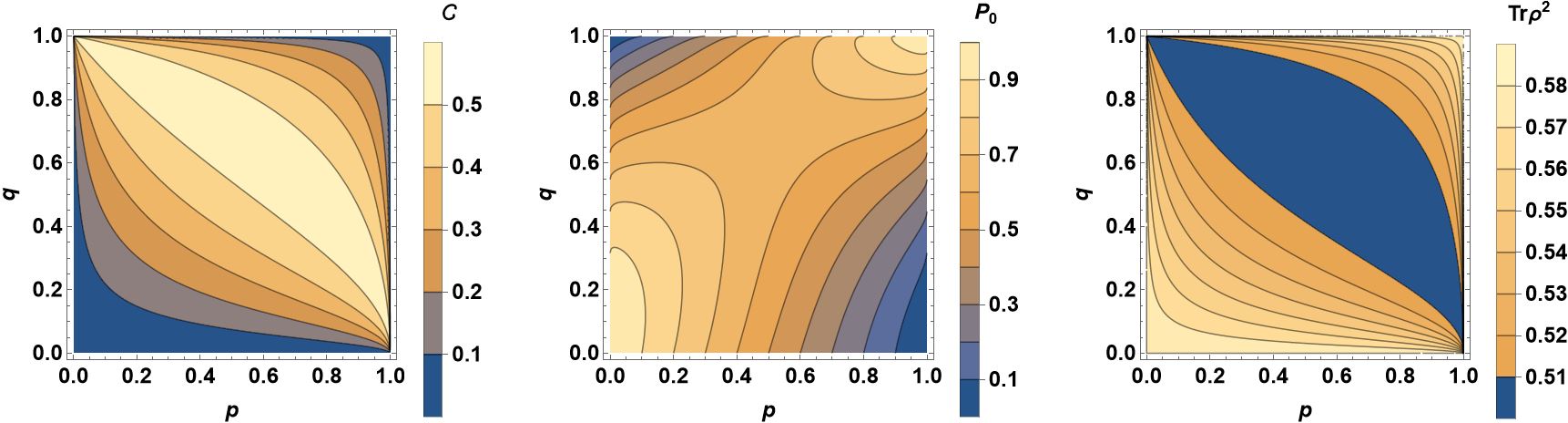}}  \\
        & (\textbf{c}) \\
        & \includegraphics[width=0.580\linewidth]{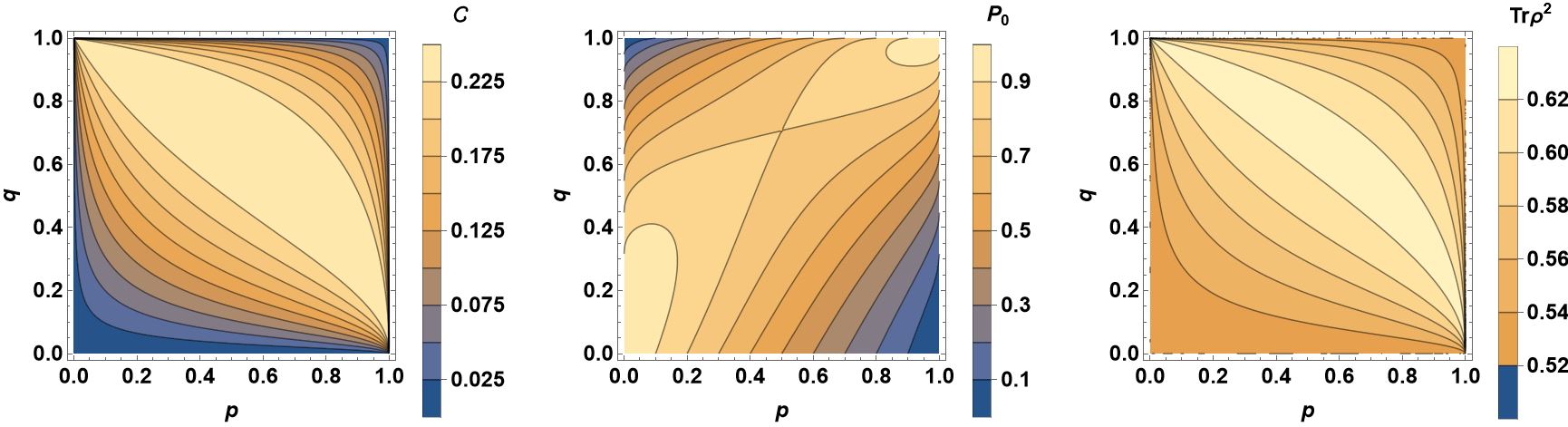} \\
        (\textbf{a}) & (\textbf{d}) \\
        \includegraphics[width=0.390\linewidth]{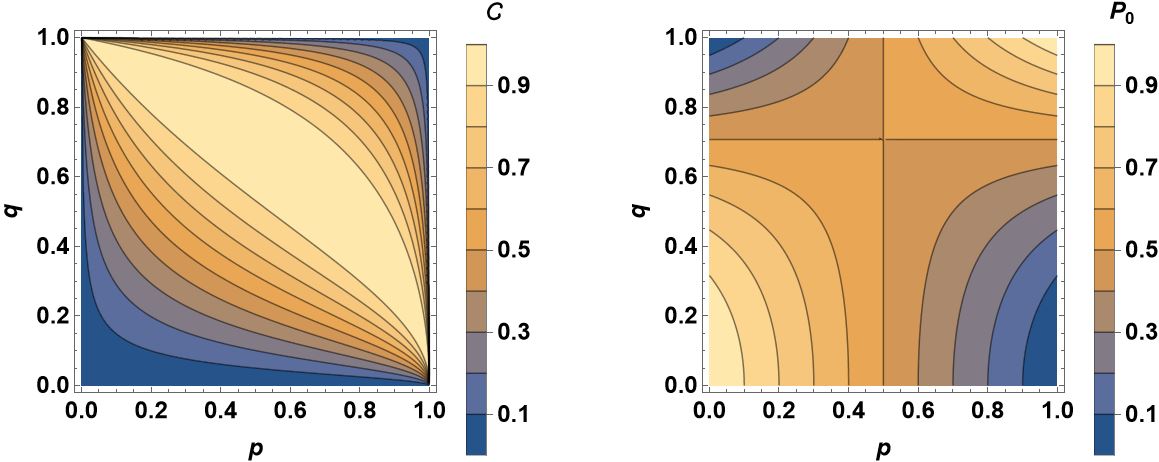} & \includegraphics[width=0.580\linewidth]{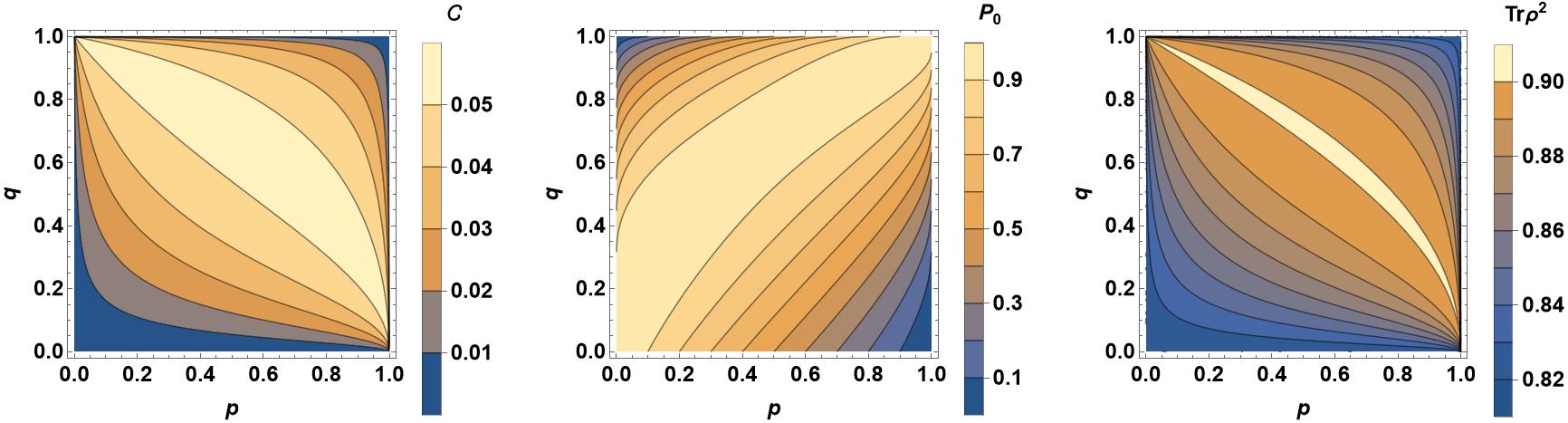} \\
        (\textbf{b}) & (\textbf{e}) \\
    \end{tabular}
\end{center}
    \caption{Concurrence $\mathcal C$ in PS for the family of channels $\alpha^1_0=1, \alpha^2_0=1-s, \alpha^2_1=s$ in colour, in agreement with each legend colour besides. (\text{a}) As a density plot function of $p,q,s$, and some contour plots of $\mathcal C$, $P_0$, and ${\rm Tr}\rho^2_{\rm out}$ together for (\textbf{b}) $s=1$ (here, ${\rm Tr}\rho^2_{\rm out}=1$ is not reported),  (\textbf{c}) $s=\frac{7}{10}$, (\textbf{d}) $s=\frac{4}{10}$, and (\textbf{e}) $s=\frac{1}{10}$.}\label{Fig2}
\end{figure}

A general double-layer density map is shown in Figure \ref{Fig2}a for $\mathcal{C}(\rho_{\rm out})$ as a function of $p,q,s$. Solid colors involve the values referred to the color legend marked by the gray region on the left side in the bar, while transparent colors comprehend the values marked by the white region in the bar. The other panels report slice contour plots for certain values of $s$. Each contour plot includes additional contour plots for $P_0$ and ${\rm Tr}\rho^2$ as functions of $p$ and $q$. Figure \ref{Fig2}b shows the case for $s=1$ (transparent channel). In this case, only separable states are obtained, thus giving ${\rm Tr}\rho^2=1$, which is then not reported. Other cases include cases (c) $s=\frac{7}{10}$, (d) $s=\frac{4}{10}$, and (e) $s=\frac{1}{10}$ reporting $\mathcal{C}(\rho),P_0,{\rm Tr}\rho^2$ together. In general, it is observed that the case $s=1$ reaches the highest concurrence values $\mathcal{C}(\rho)=1$ near the diagonal in the plane $p,q$. The concurrence values decrease near zero with the lower values of $s$ reported in the remaining panels. In general, it is also observed that $P_0$ is low (but not necessarily the lowest) for the most entanglement states because the highest values of $P_0$ are located at the corners $p,q \approx 0,1$. In addition, purity reduces for the $s$ middle values. The structure observed in Fig.~\ref{Fig2} reflects the interference origin of the generated entanglement. Maximal concurrence appears near the diagonal regions of the parameter space where the two operational branches contribute with comparable amplitudes, thereby maximizing coherent interference after post-selection. In contrast, highly asymmetric configurations suppress interference and consequently reduce entanglement generation.

The inverse relation between concurrence and success probability reveals the probabilistic nature of the mechanism. Operational configurations producing the strongest interference-induced correlations occur precisely in regions where the selected post- measurement outcomes become less likely. This trade-off highlights the stochastic character of entanglement generation under coherent control.

\subsubsection{Entanglement generated by non-coherent control schemes}

A second example is the case with $\Lambda_1$ defined by $\alpha^1_0=1$ and $\Lambda_2$ defined by $\alpha^2_i=s, i=1,2,3, s \in [0, \frac{1}{3}]$. The results are shown in Figure \ref{Fig3} equivalently to those in Figure \ref{Fig2}. Figure \ref{Fig3}a is a general view of $\mathcal{C}(\rho_{\rm out})$ as a function of $p,q,s$ geometrically identical to before, but in a range below $\frac{1}{2}$. As before, the remaining panels report slice contour plots for (b) $s=0$ (the transparent channel, where only separable states with ${\rm Tr}\rho^2=1$ are obtained), (c) $s=\frac{1}{6}$, (d) $s=\frac{1}{4}$ (depolarizing channel) and (e) $s=\frac{1}{3}$ (ICO channel \cite{delgado3}) reporting $\mathcal{C}(\rho),P_0,{\rm Tr}\rho^2$ together. A higher value (despite being limited) for the concurrence is observed with increasing $s$. The maximum value $\mathcal{C}(\rho)=\frac{1}{2}$ for $s=\frac{1}{3}$ is reached near the diagonal as before. In addition, $P_0$ is low for the most entanglement states, with the highest values of $P_0$ near the corners $p,q \approx 0,1$. Purity reduces for the lower values of $s$. 
The reduction of purity in the intermediate regions of the parameter space reflects the increasing influence of decoherence introduced by the Pauli channels. Nevertheless, coherent control still enables the activation of nonlocal correlations despite the mixed nature of the resulting states.

\begin{figure}[htb]
\begin{center}
    \begin{tabular}{cc}
        {\multirow{3}[2]{*}[20mm]{\raisebox{0.5\height}{\includegraphics[width=0.350\linewidth]{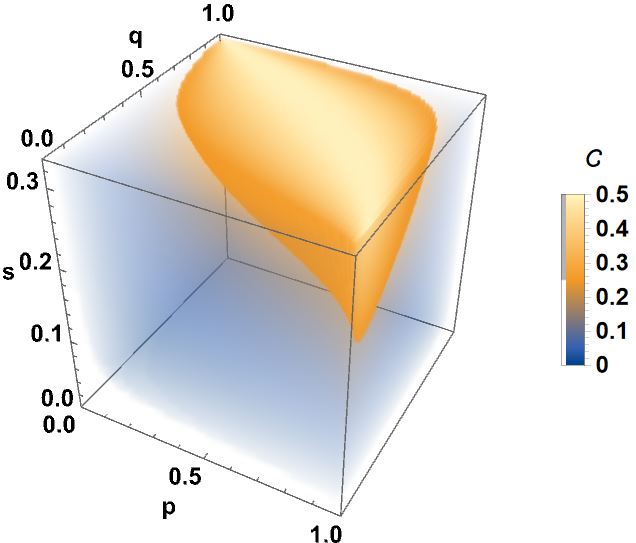}}}}  &
        {\includegraphics[width=0.580\linewidth]{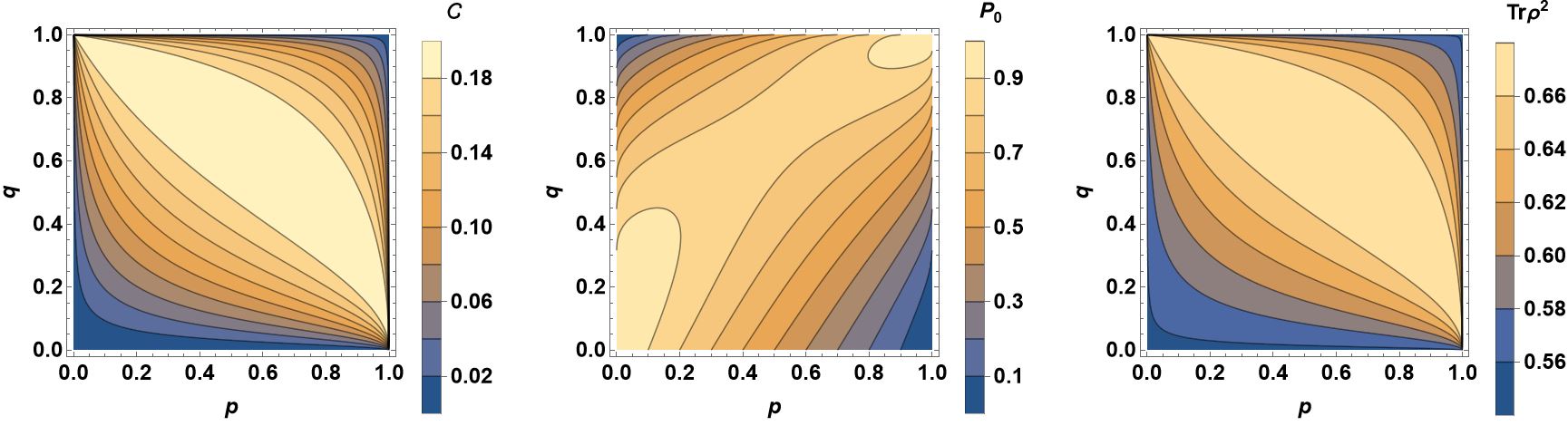}}  \\
        & (\textbf{c}) \\
        & \includegraphics[width=0.580\linewidth]{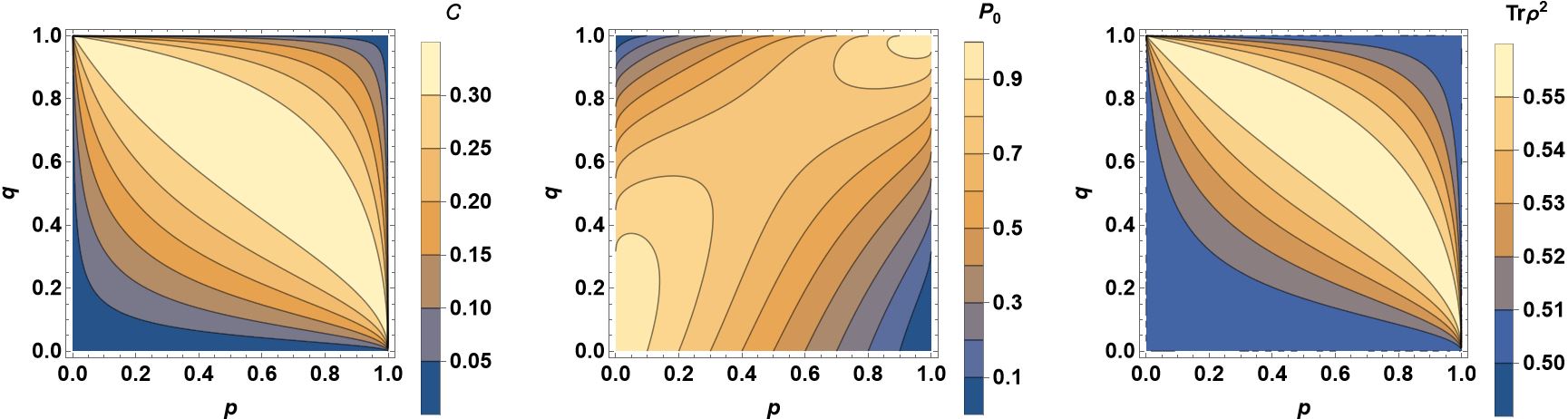} \\
        (\textbf{a}) & (\textbf{d}) \\
        \includegraphics[width=0.390\linewidth]{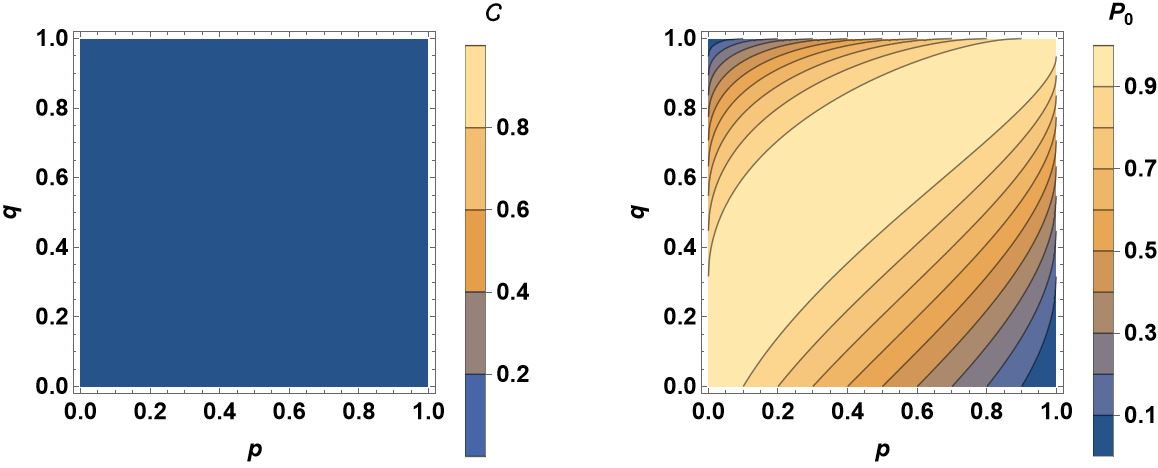} & \includegraphics[width=0.580\linewidth]{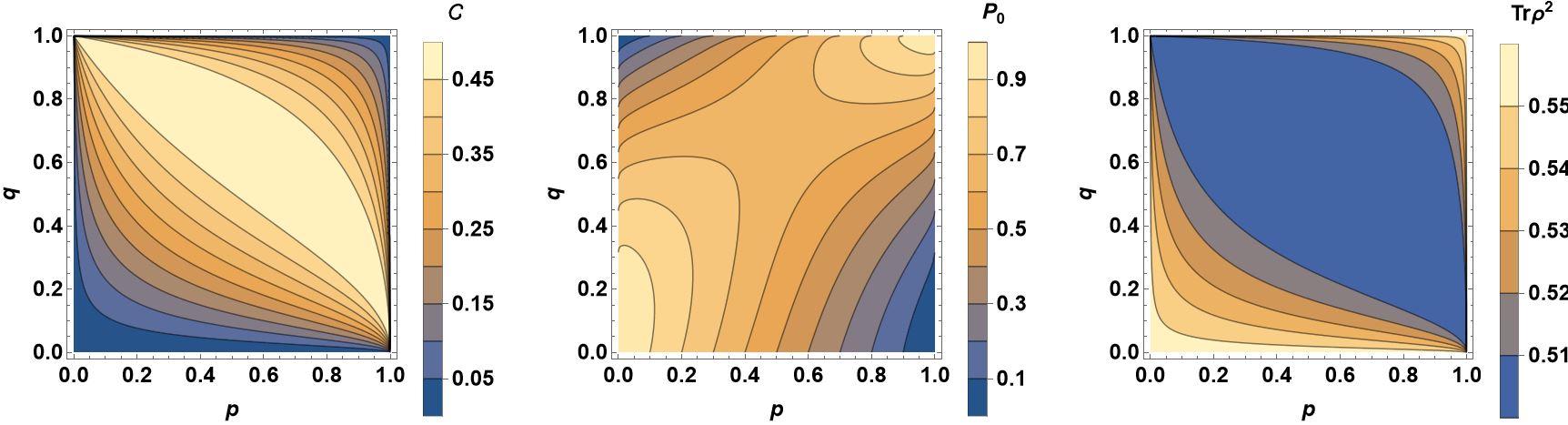} \\
        (\textbf{b}) & (\textbf{e}) \\
    \end{tabular}
\end{center}
    \caption{Concurrence $\mathcal C$ in PS for the family of channels $\alpha^1_0=1, \alpha^2_i=s, i=1,2,3, s \in [0,\frac{1}{3}]$ in colour, in agreement with each legend colour besides. (\text{a}) As a density plot function of $p,q,s$, and some contour plots of $\mathcal C$, $P_0$, and ${\rm Tr}\rho^2_{\rm out}$ together for (\textbf{b}) $s=0$ (here, ${\rm Tr}\rho^2_{\rm out}=1$ is not reported),  (\textbf{c}) $s=\frac{1}{6}$, (\textbf{d}) $s=\frac{1}{4}$, and (\textbf{e}) $s=\frac{1}{3}$.}\label{Fig3}
\end{figure}

The third case corresponds to $\Lambda_1$ defined by $\alpha^1_0=1$ and $\Lambda_2$ defined by $\alpha^2_1=s, \alpha^2_2=1-s$. The results are shown in Figure \ref{Fig4} as before. Figure \ref{Fig3}a shows the general view of $\mathcal{C}(\rho_{\rm out})$ as a function of $p,q,s$. In this case, the double-layer density map shows in solid colors the values marked by the gray marks on the left side of the legend. We find an independent behavior from $s$. The values reach $\mathcal{C}(\rho)=1$. In fact, $\rho_{\rm out}$ takes in this case the form:

\begin{figure}[thb]
\begin{center}
    \begin{tabular}{cc}
    \includegraphics[width=0.330\linewidth]{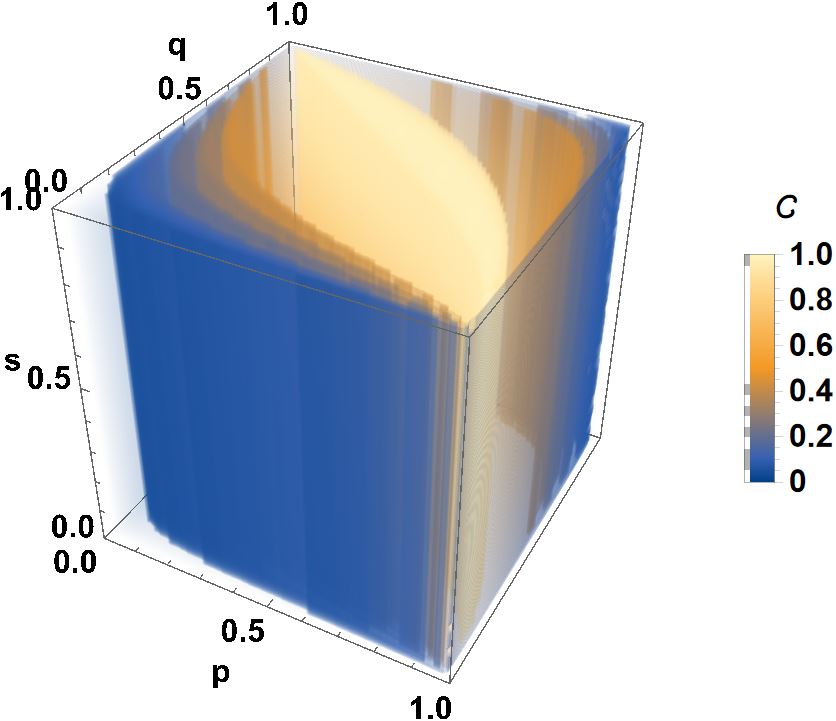} & \includegraphics[width=0.630\linewidth]{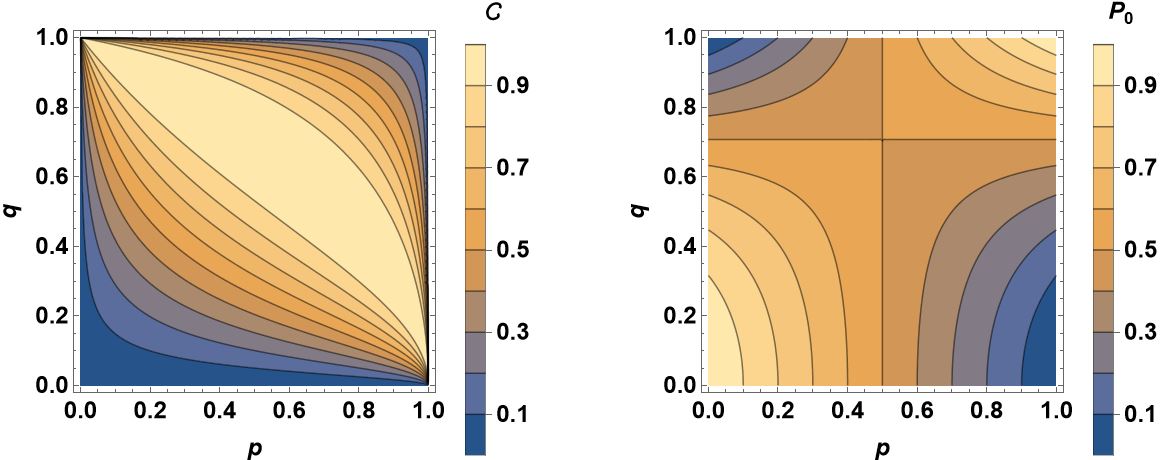} \\
    (\textbf{a}) & (\textbf{b}) 
    \end{tabular}
\end{center}
    \caption{Concurrence $\mathcal C$ in PS for the family of channels $\alpha^1_0=1, \alpha^2_1=s=1-\alpha^2_2$ in colour, in agreement with each legend colour besides. (\text{a}) As a density plot function of $p,q,s$, and the contour plots of $\mathcal C$ and $P_0$ for all $s$. ${\rm Tr}\rho^2_{\rm out}=1$ for all $s$.}\label{Fig4}
\end{figure}

\begin{eqnarray}
\rho_{\rm out}=
\left(
\begin{array}{cccc}
 0 & 0 & 0 & 0 \\
 0 & \frac{p q^2}{p \left(2 q^2-1\right)-q^2+1} & \frac{q \sqrt{(1-p) p \left(1-q^2\right)}}{p \left(2 q^2-1\right)-q^2+1} & 0 \\
 0 & \frac{q \sqrt{(1-p) p \left(1-q^2\right)}}{p \left(2 q^2-1\right)-q^2+1} & \frac{(1-p) \left(1-q^2\right)}{p \left(2 q^2-1\right)-q^2+1} & 0 \\
 0 & 0 & 0 & 0 \\
\end{array}
\right)
\end{eqnarray}

\noindent explaining the independent behavior of $s$.

Thus, Figure \ref{Fig4}b shows a slice of it for any $s$, together with the contour plot for $P_0$ whose behavior is similar to the previous cases. The states become pure for the overall values of $s$.

\subsubsection{Entangled states generated by the overall Pauli channels in their parametric space}

In agreement with the previous analysis, we propose to analyze the generation of stochastic entanglement by setting $q=\sqrt{p}$, but extending the analysis over the whole parametric space of the Pauli channels $(\alpha_1,\alpha_2,\alpha_3)$. This space has previously been used, identifying several emblematic channels (Figure \ref{Fig5}a). In this case, a direct calculation gives the following $\rho_{\rm out}$:

\begin{eqnarray}
\rho_{\rm out}=
    \left(
\begin{array}{cccc}
 \frac{1-A}{1-2A p (1-p)} & 0 & 0 & 0 \\
 0 & \frac{A p^2}{1-2A p (1-p)} & \frac{A (1-p) p}{1-2A p (1-p)} & 0 \\
 0 & \frac{A (1-p) p}{1-2A p (1-p)} & \frac{A (1-p)^2}{1-2A p (1-p)} & 0 \\
 0 & 0 & 0 & 0 \\
\end{array}
\right)
\end{eqnarray}

\noindent where $A=\alpha_1+\alpha_2$, thus becoming independent of $\alpha_3$ as our previous analysis suggested. From this expression, the following outcomes are easily obtained:

\begin{eqnarray}
\mathcal{C}(\rho_{\rm out}) &=& \frac{2A p (1-p)}{1-2A p (1-p)}\\
P_0 &=& 1-2A p (1-p) \\
{\rm Tr}\rho^2_{\rm out} &=& \frac{4 - 8A + 5A^2}{(A-2)^2}
\end{eqnarray}

Figure \ref{Fig5}b shows all previous functions together in the $p, A$ plane. Thus, $\mathcal{C}(\rho_{\rm out})$ increases as a function of the increment of $A$, but also mainly around $p=\frac{1}{2}$ (peach surface). The maximum is reached when $p=\frac{1}{2}$. Instead, $P_0$ (blue surface) decreases inversely, reducing the probability of obtaining the maximum entanglement. Finally, ${\rm Tr} \rho^2_{\rm out}$ (green surface) becomes independent of $p$, with pure states in $A=0,1$, and with the most mixed states for $A=\frac{2}{3}$. Figure \ref{Fig5}c shows the particular case for $p=\frac{1}{2}$ where the maximum entangled states are reached. 
The global structure of the Pauli-channel parameter space indicates that entanglement generation depends primarily on the combined weight of the nontrivial Pauli components rather than on the specific orientation of the channel within the parameter space. This suggests that the operational interference responsible for entanglement generation is governed mainly by the balance between identity and non-identity contributions within the controlled channels.

\begin{figure}[htb]
\begin{center}
    \begin{tabular}{ccc}
    \includegraphics[width=0.350\linewidth]{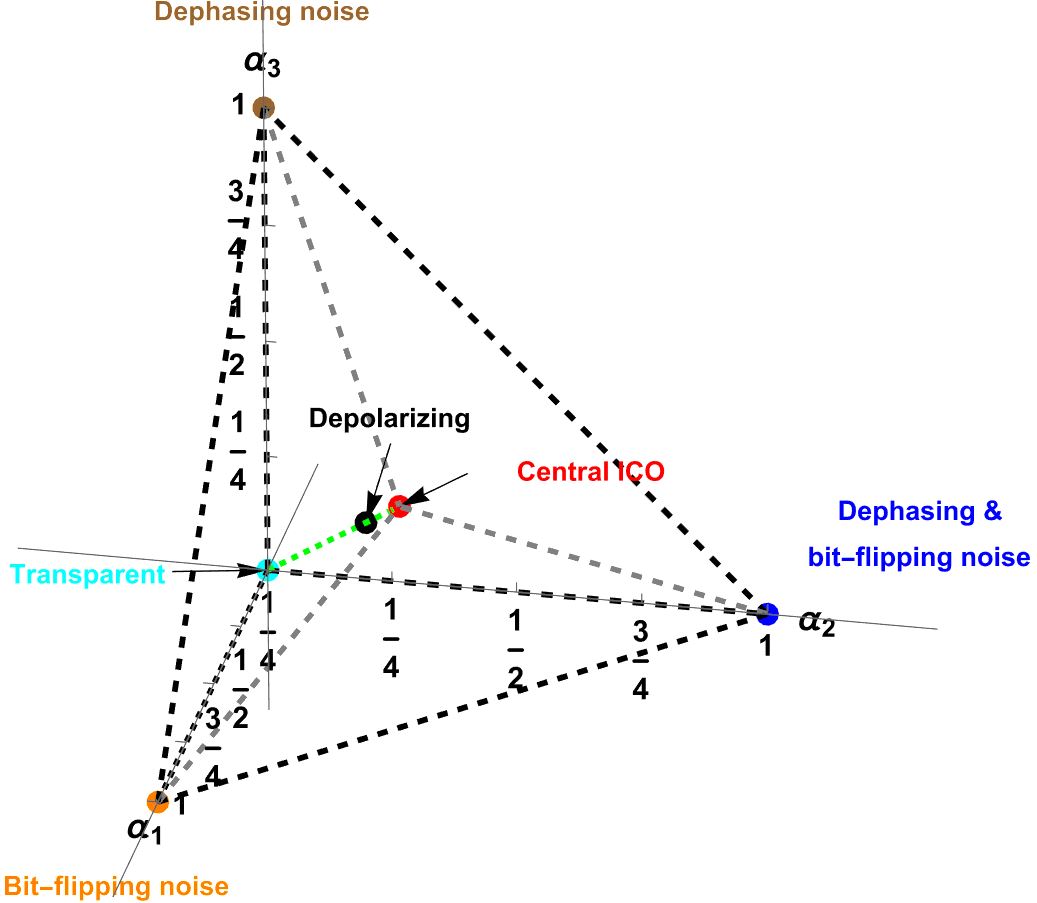} & 
    \includegraphics[width=0.255\linewidth]{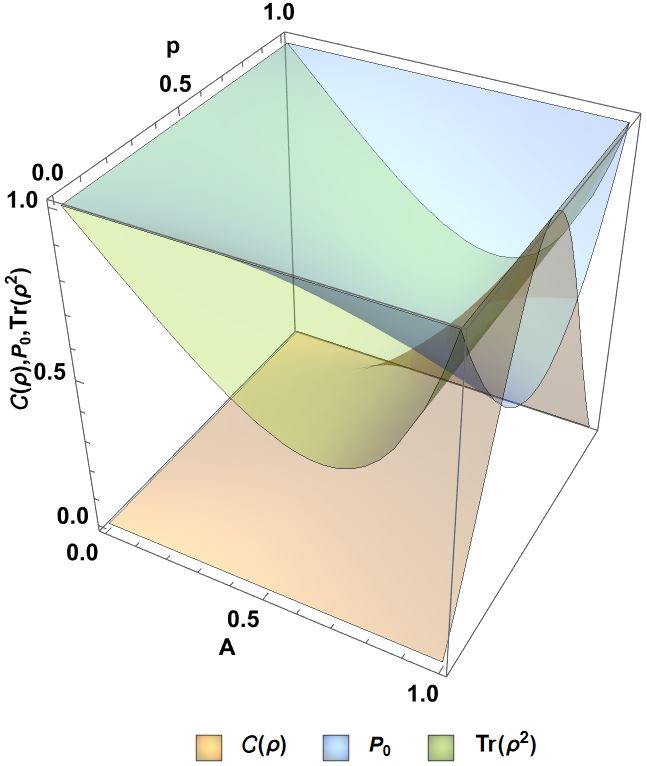} \hspace{5mm} & \includegraphics[width=0.270\linewidth]{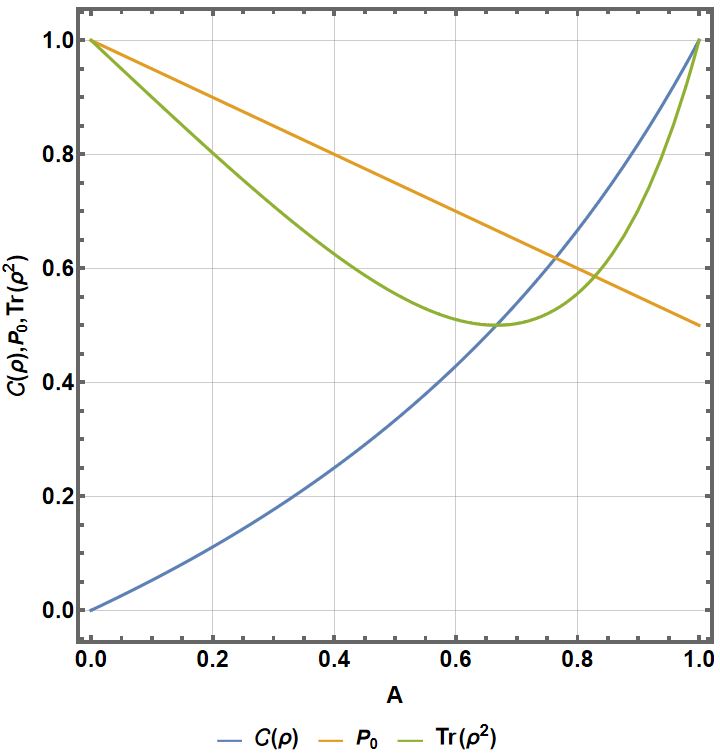} \\ 
    (\textbf{a}) & (\textbf{b}) & (\textbf{c}) 
    \end{tabular}
\end{center}
    \caption{(\text{a}) Pauli channels parametric space with some emblematic channels, (\text{b}) plots of $\mathcal{C}(\rho),P_0,{\rm Tr}\rho_{\rm out}^2$ as functions of $A$ and $p$, and (\text{c}) particular case for $p=\frac{1}{2}$.}\label{Fig5}
\end{figure}
\section{Discussions}
In standard quantum-information architectures, entanglement is commonly generated through nonlocal interactions implemented by entangling gates such as controlled-NOT or controlled-phase operations acting directly on the target subsystems \cite{nielsen1}. In contrast, the present framework employs only local operations acting independently on each subsystem. Entanglement emerges from coherent interference between alternative operational pathways selected by a control system and revealed after post-selection.
From this perspective, the generated entanglement should be interpreted as an interference phenomenon rather than as a direct consequence of subsystem interaction. Each operational branch individually preserves separability, while coherent superposition between branches enables the appearance of nonlocal correlations after measurement of the control system. This establishes a direct connection between coherent control and interference-induced entanglement generation.

The present approach is related to previous works involving coherent control of quantum channels and indefinite causal order, where superposition of operational structures has been shown to enhance communication and information-processing tasks. In contrast, the main objective here is the direct generation of entanglement from initially separable states through coherent superposition of local operations. The results obtained in the unitary and noisy scenarios further indicate that the same operational principle persists even in the presence of decoherence.
The analysis of Pauli-channel architectures reveals a nontrivial interplay between concurrence, purity, and success probability. Maximal entanglement generation occurs in parameter regimes where coherent interference between alternative evolutions is strongest, although these regions generally correspond to lower post-selection probabilities. Despite the presence of noise, coherent control of alternative operational pathways still enables the activation of nonlocal correlations from separable mixed states.

Overall, these results suggest that coherent superposition of operational structures may itself constitute a generalized resource for entanglement generation. This perspective unifies the deterministic and stochastic regimes analyzed throughout this work under a common mechanism based on interference between alternative quantum evolutions. Such coherent-control architectures may also provide new possibilities for distributed quantum-information processing and programmable quantum communication protocols.

\section{Conclusions}\label{conclu}
In this work, we have analyzed the generation of quantum entanglement through both coherent and non-coherent control strategies, focusing on architectures based on path superposition and related quantum communication frameworks. We first demonstrated that the coherent superposition of alternative sets of local unitary operations enables the deterministic generation of maximally entangled states from fully separable inputs. By explicitly deriving necessary and sufficient conditions for local operators, we established a general procedure to obtain Bell, GHZ, and W states, showing that the resulting states are locally unitary equivalent to the standard representatives of these entanglement classes.
We then extended the analysis to realistic noisy scenarios by considering pairs of Pauli channels arranged under path superposition and indefinite causal order. The controlled nature of these communication architectures introduces the interesting idea of exploiting distant imperfect resources to produce quality entangled resources to be used in a process comprising quantum processing or for quantum cryptography purposes. In this context, we derived closed analytical expressions for the output states obtained from separable mixed inputs and showed that entanglement can emerge in a fundamentally stochastic manner. Using concurrence as a quantitative measure, we identified parameter regimes in which entanglement is generated despite the presence of noise, revealing that quantum control over the ordering or selection of channels can activate correlations that are not accessible through conventional fixed-order implementations. {In particular, these imperfect communication structures that introduce Pauli noise can produce entanglement under a controlled process at distance.

Our results further highlight the interplay between entanglement, success probability, and purity. In particular, we observed that higher levels of entanglement are generally associated with lower success probabilities, and that noise introduces nontrivial trade-offs between coherence and mixedness. By exploring representative families and the full parametric space of Pauli channels, we provided a global characterization of these trade-offs, identifying the operational conditions under which stochastic entanglement generation is optimized.

Overall, this work shows that both coherent and non-coherent control paradigms offer powerful and complementary mechanisms for entanglement generation beyond standard circuit-based approaches. These findings contribute to the understanding of quantum resources under generalized control architectures and may have practical implications for quantum communication, distributed quantum computing, and emerging quantum network technologies.

\section*{Acknowledgements}

The support of the School of Engineering and Sciences of Tecnologico de Monterrey is acknowledged. ME thanks Prof. Karol \.Zyczkowski for fruitful discussions.


\begin{thebibliography}{999}

\bibitem{nielsen1}
Nielsen, M.A.; Chuang, I.L. \emph{Quantum Computation and Quantum Information}; {Cambridge University Press}: {Cambridge, UK},   {2010}.

\bibitem{bengtsson2017}
Bengtsson, I.; Życzkowski, K. \emph{Geometry of Quantum States: An Introduction to Quantum Entanglement}; {Cambridge University Press}: {Cambridge, UK}, {2017}.

\bibitem{yang2022}
Yang, Z.-B.; Wang, Y.-P.; Li, J.; Hu, C.-M.; You, J. Q. Entanglement emerges from dissipation-driven quantum self-organization. {\em J. Magn. Magn. Mater.} {\bf 2022}, \emph{564}, 170139.

\bibitem{wu2017}
Wu, J.-L.; Ji, X.; Zhang, S. Fast adiabatic quantum state transfer and entanglement generation between two atoms via dressed states. {\em Sci. Rep.} {\bf 2017}, \emph{7}, 46255.

\bibitem{bhattacharya2023}
Bhattacharya, U.; Lamprou, T.; Maxwell, A. S.; Ordóñez, A. F.; Pisanty, E.; Rivera-Dean, J.; Stammer, P.; Ciappina, M. F.; Lewenstein, M.; Tzallas, P. Strong-laser-field physics, non-classical light states and quantum information science. {\em Rep. Prog. Phys.} {\bf 2023}, \emph{86}, 094401.

\bibitem{lamprou2024}
Lamprou, T.; Stammer, P.; Rivera-Dean, J.; Tsatrafyllis, N.; Ciappina, M. F.; Lewenstein, M.; Tzallas, P. Recent developments in the generation of non-classical and entangled light states using intense laser–matter interactions. {\em arXiv} {\bf 2024}, arXiv:2410.17452.

\bibitem{gour2025}
Gour, G. Mixed-state entanglement. In \emph{Quantum Resource Theories}; 
{Cambridge University Press}: {Cambridge, UK}, {2025}; pp. 538--623.

\bibitem{Horodecki1998}
Horodecki, M.; Horodecki, P.; Horodecki, R. 
Mixed-State Entanglement and Distillation: Is There a Bound Entanglement in Nature? 
{\em Phys. Rev. Lett.} {\bf 1998}, {\em 80}, 5239--5242. 

\bibitem{Horodecki2000}
Horodecki, M.; Horodecki, P.; Horodecki, R. 
Asymptotic Manipulations of Entanglement Can Exhibit Genuine Irreversibility. 
{\em Phys. Rev. Lett.} {\bf 2000}, {\em 84}, 4260--4263. 

\bibitem{Horodecki2005}
Horodecki, K.; Horodecki, M.; Horodecki, P.; Oppenheim, J. 
Secure Key from Bound Entanglement. 
{\em Phys. Rev. Lett.} {\bf 2005}, {\em 94}, 160502. 

\bibitem{yamasaki2022}
Yamasaki, H.; Morelli, S.; Miethlinger, M.; Bavaresco, J.; Friis, N.; Huber, M. 
Activation of genuine multipartite entanglement: Beyond the single-copy paradigm of entanglement characterisation. 
{\em Quantum} {\bf 2022}, \emph{6}, 695.

\bibitem{rinp2022}
Singh, U.; Bandyopadhyay, S.; Adhikari, S. 
Generation of entanglement in mixed states via quantum operations. 
{\em Results Phys.} {\bf 2022}, \emph{40}, 105830.

\bibitem{moharramipour2024}
Moharramipour, A.; Lessa, L. A.; Wang, C.; Hsieh, T. H.; Sahu, S. 
Symmetry-Enforced Entanglement in Maximally Mixed States. 
{\em PRX Quantum} {\bf 2024}, \emph{5}, 040336.

\bibitem{chiribella2013}
Chiribella, G.; D'Ariano, G. M.; Perinotti, P.; Valiron, B. 
Quantum computations without definite causal structure. 
{\em Phys. Rev. A} {\bf 2013}, \emph{88}, 022318.

\bibitem{rubino2017}
Rubino, G.; Rozema, L. A.; Feix, A.; Araújo, M.; Zeuner, J. M.; Procopio, L. M.; Brukner, Č.; Walther, P. 
Experimental verification of an indefinite causal order. 
{\em Sci. Adv.} {\bf 2017}, \emph{3}, e1602589.

\bibitem{Ebler2018}
Ebler, D.; Salek, S.; Chiribella, G. Enhanced communication with the assistance of indefinite causal order. \em{Phys. Rev. Lett.} \textbf{2018}, \em{120}, 120502.

\bibitem{Procopio2019}
Procopio, L.M.; Delgado, F.; Enr\'{i}quez, M.; Belabas, N.; Levenson, J.A. Communication enhancement through quantum coherent control of N channels in an indefinite causal-order scenario. \em{Entropy} \textbf{2019}, \em{21}, 1012.

\bibitem{Bavaresco2021}
Bavaresco, J.; Murao, M.; Quintino, M.T. Unitary channel discrimination beyond group structures: Advantages of sequential and indefinite-causal-order strategies. {\em arXiv} \textbf{2021}, arXiv:2105.13369.

\bibitem{Zhao2020}
Zhao, X.; Yang, Y.; Chiribella, G. Quantum metrology with indefinite causal order. \emph{Phys. Rev. Lett.} \textbf{2020}, \emph{124}, 190503.

\bibitem{DelSanto2024}
Del Santo, F.; Ebler, D.; Chiribella, G. Noisy quantum parameter estimation with indefinite causal order. \emph{Phys. Rev. A} \textbf{2024}, \emph{109}, 012603.

\bibitem{Abbott2018}
Abbott, A.A.; Wechs, J.; Horsman, C.; Mhalla, M.; Branciard, C. Communication through coherent control of quantum channels. \em{Quantum} \textbf{2018}, \em{2}, 76.

\bibitem{Chiribella2019}
Chiribella, G.; Kristj\'ansson, H. Quantum Shannon theory with superpositions of trajectories. \em{Proc. R. Soc. A} \textbf{2019}, \em{475}, 20180903.

\bibitem{Rubino2021}
Rubino, G.; Rozema, L.A.; Ebler, D.; Kristj\'ansson, H.; Salek, S.; Gu\'erin, P.A.; Abbott, A.A.; Branciard, C.; Brukner, \v{C}.; Chiribella, G.; Walther, P. Experimental quantum communication enhancement by superposing trajectories. \em{Phys. Rev. Research} \textbf{2021}, \em{3}, 013093.

\bibitem{Delgado2023}
Delgado, F. Parametric symmetries in architectures involving indefinite causal order and path superposition for quantum parameter estimation of Pauli channels. \em{Symmetry} \textbf{2023}, \em{15}, 1097.

\bibitem{Mondal2025}
Mondal, S.; Ghosh, P.; Sen, U. Path superposition as a resource for perfect quantum teleportation with separable states. {\em arXiv} \textbf{2025}, arXiv:2505.11398.

\bibitem{CastanosCervantes2024}
Casta\~nos-Cervantes, L.O.; Procopio, L.M.; Enr\'iquez, M. Coherent control of two Jaynes--Cummings cavities. \em{Sci. Rep.} \textbf{2024}, \em{14}, 3790.

\bibitem{sorensen2003}
Sørensen, A. S. and Mølmer, K. Measurement Induced Entanglement and Quantum Computation with Atoms in Optical Cavities, {\em Phys. Rev. Lett.} {\bf 2003}, {em 91}(9), 097905.

\bibitem{kim2012}
Kim, Y. S., Lee, J. C., Kwon, O. et al. Protecting entanglement from decoherence using weak measurement and quantum measurement reversal, {\em Nature Phys} {\bf 2012}, {\em 8}, 117–120.

\bibitem{white2016}
White, T., Mutus, J., Dressel, J. et al. Preserving entanglement during weak measurement demonstrated with a violation of the Bell–Leggett–Garg inequality. {\em npj Quantum Inf.} {\bf 2016},  {\em 2}, 15022.

\bibitem{Grimaudo2020}
Grimaudo, R., Messina, A., Sergi, A., Vitanov, N. V., Filippov, S. N. Two-Qubit Entanglement Generation through Non-Hermitian Hamiltonians Induced by Repeated Measurements on an Ancilla. {\em Entropy} {\bf 2020}, {\em 22}(10), 1184.

\bibitem{koudia2023}
Koudia, S.; Cacciapuoti, A. S.; Caleffi, M. Deterministic generation of multipartite entanglement via causal activation in the quantum internet. {\em IEEE Access} {\bf 2023}, \emph{11}, 73863--73878.

\bibitem{flammia2020}
Flammia, S. T. and Wallman, J. J., Efficient Estimation of Pauli Channels, {\emph ACM Transactions on Quantum Computing}, {\bf 2020}, {\emph 1}(1), 32.

\bibitem{delgado3}
Delgado, F.; Cardoso-Isidoro, Carlos. Performance characterization of Pauli channels assisted by indefinite causal order and post-measurement. {\em Quantum Inf. Comput.}  {\bf 2020}, \emph{20}, 1261--1280.

\bibitem{kraus1}
Kraus, K. {\em States, Effects and Operations: Fundamental Notions of Quantum Theory}; Springer: Berlin, {Germany}, 1983. 

\bibitem{HW97}
Hill, S. A.; Wootters, W. K. Entanglement of a pair of quantum bits. {\em Phys. Rev. Lett.} {\bf 1997}, {\em 78}, 5022.

\bibitem{W98}
Wootters, W.K. Entanglement of formation for any arbitrary state of two qubits, {\em Phys. Rev. Lett.} {\bf 1998}, {\em 80}, 2245.

\end{thebibliography}
\end{document}